\documentclass[conference]{IEEEtran}
\usepackage{caption}
\usepackage{amsthm,amssymb,graphicx,graphicx,multirow,amsmath,color,amsfonts}
\usepackage{tikz}
\usepackage{pgfplots}
\usetikzlibrary{shapes,shadows,arrows}

\def\E{\mathbb{E}}

\def\P{\mathbb{P}}

\def\i{\mathbf{1}}

\def\d{\mathrm{d}}

\newtheorem{lemma}{Lemma}{}

\author{
\IEEEauthorblockN{Radha Krishna Ganti, Andrew Thangaraj and Arijit Mondal\\}
\IEEEauthorblockA{Department of Electrical Engineering\\
Indian Institute of Technology, Madras\\
Chennai, India 600036\\
\{rganti, andrew, arijit.mondal\}@ee.iitm.ac.in}
}

\def\diag{\operatorname{diag}}
\def\trace{\operatorname{tr}}

\def\ent{H}
\def\EN{\mathcal{E}}
\def\cQ{\mathcal{Q}}

\title{Approximation of Capacity for ISI Channels with One-bit Output Quantization}
\date{\today}
\begin{document}
\maketitle
\begin{abstract}
Motivated by recent high bandwidth communication systems, Inter-Symbol
Interference (ISI) channels with 1-bit quantized
output are considered under an average-power-constrained continuous
input. While the exact capacity is difficult to characterize, an
approximation that matches with the exact channel
output up to a probability of error is provided. The
approximation does not have additive noise, but constrains the
channel output (without noise) to be above a threshold in absolute
value. The capacity under the approximation is computed using methods involving standard
Gibbs distributions. Markovian achievable schemes approaching the
approximate capacity are provided. The methods used over the
approximate ISI channel result in ideas for practical coding schemes
for ISI channels with 1-bit output quantization.
\end{abstract}
\section{Introduction} 
Channels with Inter-Symbol Interference (ISI) and Additive White
Gaussian Noise (AWGN) are often encountered in practice. Depending on
the application, an average-input power constraint or a finite input
alphabet constraint is commonly studied. Recently, in applications
such as millimeter wave \cite{SunHeath14}\cite{AlkhateebMoHeath14} or optical
or intra-chip \cite{Harwood07} communications, quantization of the
output of an ISI channel has been considered because of limitations in
Analog-to-Digital conversion at high speeds. In some cases, the output
quantization may be as low as a single bit. Since the transmitters in
some of these systems can be more complex and operate at high powers, the channel input may not
have severe quantization limits.

Motivated by the above applications, we consider a noisy ISI channel
with average-power constrained continuous input and 1-bit quantized
output. The available literature mostly considers either continuous
input/output alphabets or a finite input alphabet with a
continuous output alphabet
\cite{ShamaiLaroia96}\cite{SadeghiVontobel09}. The quantized output
case has been considered for the AWGN channel with no ISI
\cite{SinghDabeerMadhow09}, and the ISI case has been briefly
addressed recently \cite{MoHeath14}.

The exact capacity of an ISI channel with 1-bit quantized output
appears to be difficult to characterize explicitly. In this work, we introduce an approximation to the ISI channel model with
1-bit quantized output. The
approximation does not have additive noise, but constrains the
channel output (without noise) to be above a threshold in absolute
value. Because of the thresholding, the
approximate channel output, after quantization, matches the actual channel output up to a
probability of error that can be controlled by the threshold. 

The main
advantage of avoiding noise in the approximation is that the exact
capacity can be computed for the approximate model. We show how such a
computation can be carried out using Gibbs
distributions. In addition, we exhibit achievable schemes with Markov input
approaching the approximate capacity in some numerical examples. Since the approximation is valid up to a probability of error, a coding scheme used over the approximate channel model can be
coupled with a standard error control code to derive a practical
coding scheme for the exact channel. 
  
The rest of the paper is organized as follows. Section
\ref{sec:system-model} describes the ISI channel model and its approximation. In Section
\ref{sec:capacity-channel}, we provide a method for computing the
approximate capacity under an average power constraint on the input, and elaborate on achievable
schemes in Section \ref{sec:achievable-schemes}. Numerical results are
given in Section \ref{sec:numerical-examples}, and concluding remarks
are made in Section \ref{sec:concluding-remarks}.

\section{System Model}
\label{sec:system-model}
We consider a discrete-time finite-tap ISI channel with average-power-constrained continuous input and one-bit quantized output as depicted in Fig. \ref{fig:isimodel}.
\begin{figure}[htb]
  \centering
\begin{tikzpicture}
\node (in) {};
\node[draw,right of=in,xshift=1cm,align=center] (ch) {Channel\\$h$}
edge[<-] node[above]{$X_n$} (in);
\node[draw,circle,right of=ch,xshift=1cm] (plus) {+}
edge[<-] node[above]{$Y_n$} (ch);
\node[above of=plus] (awgn) {$Z_n$}
edge[->] (plus);
\node[draw,right of=plus,xshift=1cm] (Q) {$Q(\cdot)$}
edge[<-] node[above]{$R_n$} (plus);
\node[right of=Q,xshift=1cm] (out) {}
edge[<-] node[above]{$Q(R_n)$} (Q);
\end{tikzpicture}  
  \caption{ISI channel with quantized output.}
  \label{fig:isimodel}
\end{figure}
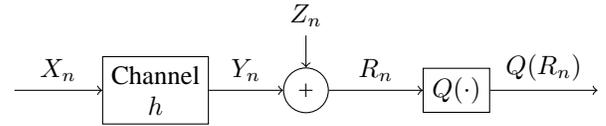
The input to the channel is denoted $X=\{X_n, 0\le n\le N-1\}$, and the channel impulse response of length $L$ is denoted $h=\{h_n, 0\leq n\leq L-1\}$. The convolution of the input with the channel impulse response is denoted $Y=\{Y_n, 0\le n\le N-1\}$, and is given by
\begin{equation}
  \label{eq:conv}
  Y_n=\sum_{k=0}^{L-1}h_kX_{n-k}.
  \end{equation}
The channel $h$ is assumed to be constant, and all signals are assumed to be zero outside the specified ranges. Independent and identically distributed zero-mean Gaussian noise of variance $\sigma^2$, denoted $Z_n$, is added to obtain an intermediate signal $R_n=Y_n+Z_n$. The signal $R_n$ is quantized by a 1-bit quantizer $Q(\cdot)$ to obtain the channel output $Q(R)=\{Q(R_n), 0\le n\le N+L-2\}$. The quantizer is defined as follows:
\begin{equation}
  \label{eq:quantizer}
  Q(x)=\begin{cases}
+1,\text{ if }x\ge0,\\
-1,\text{ if }x<0.
\end{cases}
\end{equation}
The average power of the input is constrained to be at most $P$. That is, we require
\begin{equation}
  \label{eq:powerconstraint}
  \E[\|X\|^2]=\sum_{n=0}^{N-1}E[|X_n|^2]\le NP.
\end{equation}
In this work, the overall goal is to approximate the mutual
information rate $\frac{1}{N}I(X;Q(R))$ and provide computable
expressions or bounds. 

The approximate ISI channel model is depicted in Fig. \ref{fig:approxisi}.
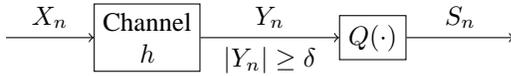
\begin{figure}[htb]
  \centering
\begin{tikzpicture}
\node (in) {};
\node[draw,right of=in,xshift=1cm,align=center] (ch) {Channel\\$h$}
edge[<-] node[above]{$X_n$} (in);
\node[draw,right of=ch,xshift=2cm] (Q) {$Q(\cdot)$}
edge[<-] node[above]{$Y_n$} node[below] {$|Y_n|\ge\delta$} (ch);
\node[right of=Q,xshift=1cm] (out) {}
edge[<-] node[above]{$S_n$} (Q);
\end{tikzpicture}    
  \caption{Approximate ISI channel with quantized output.}
  \label{fig:approxisi}
\end{figure}
In the approximate ISI channel, there is no noise. However, the convolution output $Y_n$ is constrained to be greater than $\delta$ in absolute value, and this provides justification for ignoring noise. We readily see that, under the constraint $|Y_n|\ge\delta$, the output of the actual model $Q(R_n)$ is approximated by $S_n=Q(Y_n)$ up to a probability of error lesser than or equal to $\cQ(\delta/\sigma)$, where $\cQ(x)=\int_x^{\infty}\frac{1}{\sqrt{2\pi}}e^{-u^2/2}du$ is the standard $Q$-function. Therefore, coding schemes developed for the approximate model can be used in the actual model with additional error control coding for the approximation error $\cQ(\delta/\sigma)$. While the additional coding results in a loss in information rate, we find that the approximation, because of the removal of noise, is useful on the following two counts:
\begin{enumerate}
\item The capacity or maximum information rate $\frac{1}{N}I(X;S)$, where $S=\{S_n,0\le n\le N-1\}$, under the power
  constraint of (\ref{eq:powerconstraint}) and the constraint
  $|Y_n|\ge\delta$, has computable expressions and bounds. This
  provides useful estimates of the capacity of output-quantized ISI channels.
\item The techniques used for computing the approximate capacity
  provide useful intuition on coding and signaling methods for output-quantized ISI
  channels. 
\end{enumerate}

\noindent{\em Convolution}: We will use matrix notation to denote the convolution
in \eqref{eq:conv} as $Y=\tilde{M}_hX$, where the entries of the
$N\times N$ matrix $\tilde{M}_h$ are either 0 or one of the channel
taps $h_k$ and $X$, $Y$ are column vectors. The matrix $\tilde{M}_h$ is not circulant. Let $M_h$
denote the $N\times N$ circulant matrix with first column equal to $h$. When
$N$ becomes large (as is in our case), with $L$ fixed,  $\tilde{M}_h$  behaves
like the circulant matrix $M_h$ in the sense that $\lim_{N \to \infty} \|
M_h-\tilde{M}_h\| =0$ where $\|. \|$ is a matrix norm. In this
paper\footnote{For all $N$, the
  results hold with a suitable cyclic prefix, for
  instance.}, for simplicity, we will always
assume the circular convolution $Y= M_hX$ in the channel model. 

\section{Approximate Capacity}
\label{sec:capacity-channel}
The capacity of the approximate ISI channel is given by
\begin{eqnarray*}
C_{\delta}(P)=\lim_{N \to \infty} \sup_{\stackrel{\E[\|X\|^2]\le NP}{|Y_n|\ge\delta}}\frac{I(X,S)}{N}=\frac{\ent(S)}{N},  
\end{eqnarray*}
where the last equality is because the sequence $S$ is  a deterministic function of the input $X$
in the absence of noise in the approximate channel. 
Since the output alphabet of the quantizer is either $1$ or $-1$, it is easy to observe that $C_{\delta}(P)\leq 1$.  
 
% Since  the  absolute value of the output of the   channel $Y_n$ is
% always greater than $\delta$, the input $X_n$ cannot take arbitrary
% values due to the relation $Y= HX$. In addition, fixing the  quantized
% outputs, \ie, $S_n$  constraints the input symbols. We are interested
% in the maximum achievable entropy rate of the stochastic process $S_n$
% which is given by
% \[H(S) = \lim_{N\to\infty }\frac{H(S_1,S_2,\hdots,S_N)}{N},\]
% which depends on the joint distribution of $S=\{S_1,S_2,\hdots,S_N\}$.
% We begin by quantifying the joint PMF of $S$ which   is a discrete  random variable with support $\{-1,1\}^N$. We will also assume that $\kappa$ is  large. 
% We will also assume that the filter state is zero at the first time slot. So the sequence $S_k$ would correspond to some  input of the form 
%\[X_\kappa \triangleq \{x[1], x[2], \hdots x[\kappa]\}.\]% +1-L],\underbrace{ 0\hdots 0}_{L-1}\}.\]
%{\color{red} This input would  effect $L-1$ future outputs. However we can safely neglect this as $\kappa$ can be made arbitrarily large compared to $L$. }
 
\subsection{Power constraint}
We begin by bounding the power of the input sequence $X$ required for a given output sequence $S$. 
%  In particular, it is easy to observe that if $S_n=1$ and
 % $S_{n+1}=-1$, the input alphabet $X_n$ have to change by a larger
 % value, incurring more power compared to the case when
 % $S_n=S_{n+1}=1$. So 
Given the output symbol sequence $S=s\in\{-1,1\}^N$, we have the constraint that $|Y_n|=s_n
Y_n \geq \delta$.  So, the minimum energy, denoted $\EN(s)$, required for a given
quantized sequence $s$ is given by the following optimization problem:
\begin{align}
\EN(s)=\min_{\diag(s)M_hx \succeq \delta\i} \sum_{n=0}^{N-1}|x_n|^2,
\label{eq:conv_ineq}
\end{align}
where $\diag(s)$ is an $N\times N$ diagonal matrix with $s$ on the
main diagonal and $\i$ denotes the all-1 column. The inequalities $\diag(s)M_hx \succeq \delta\i$ are
linear and the feasible space for $x$ is the intersection of hyperplanes and,
hence, convex.  So, the above optimization problem is a convex
optimization problem in $N$ variables that essentially finds
the closest point from the origin to the convex set
$\{\diag(s)M_hx \succeq \delta\i\}$.
\label{sec:ent}
 We have 
\begin{eqnarray*}
\E [\|X\|^2 ] &=& \sum_{s  \in \{-1,1\}^N}\P(S=s)\E[\|X\|^2 \ | \ S=s],\\
&\ge&\sum_{s  \in \{-1,1\}^N}\P(S=s)\EN(s)
\end{eqnarray*}
using \eqref{eq:conv_ineq}. Because of the average power constraint on
the input $X$, we have 
\begin{align}
\sum_{s  \in \{-1,1\}^N }\P(S=s)\EN(s) \leq NP.
\label{eq:cons}
\end{align}
\subsection{Entropy maximization and Gibbs distribution} 
\label{sec:gibbs}
Let $\EN_{\min}=\min_s\EN(s)$, $\EN_{\max}=\max_s\EN(s)$,
$\overline{\EN}=\frac{1}{2^N}\sum_{s}\EN(s)$. Since the
constraints \eqref{eq:cons} are linear on the probabilities, it is
well-known \cite{jaynes1957information} that the Gibbs distribution
maximizes the entropy $\ent(S)$ for
$\EN_{\min}\le NP\le \EN_{\max}$.
%We will focus on computing the maximum entropy when  $P\leq
%\mean(\EN(S)/N)$.  
The optimal distribution is the Gibbs distribution given by
\begin{align}
\P(S=s)= \frac{e^{-\frac{\beta  \EN(s)}{N} }}{Z },  s  \in \{-1,1\}^N,
\label{eq:gibbs}
\end{align}
where $Z $ is the normalizing constant, and $\beta$ is the unique
value for which \eqref{eq:cons} is met with equality.  
 The maximum entropy is given by
\begin{align}
\ent(S) &= \beta   P +  \ln(Z )\\
                           &= \beta   P  - \ln\left(\sum_{s  \in \{-1,1\}^N} e^{-\frac{\beta  \EN(s)}{N} }  \right).
                           \label{eq:gibbs}
\end{align}
For $NP< \EN_{\min}$, there exists no probability distribution that
satisfies \eqref{eq:cons}. It is also known
\cite{jaynes1957information}\cite{Conrad_probabilitydistributions}
that $\beta=0$ when $NP=\overline{\EN}$ and we observe that the
corresponding Gibbs distribution is the uniform distribution on
$\{-1,1\}^N$ and the maximum $\ent(S)=N$. For $P= \EN_{\min}$, we have $\beta=\infty$ and the
maximum entropy is given by $\log_2(|\mathcal{S}_m|)$, where
$|\mathcal{S}_m|$ denotes the number of sequences in $\{-1,1\}^N$ that
achieve the minimum energy. Hence, the maximum entropy when $P=
\EN_{\min}$ is $\ent(S)=\frac{\log_2(|\mathcal{S}_m|)}{N}$. 

In summary, we see that the capacity of the approximate ISI channel is
given by the Gibbs distribution whenever the power constraint is above
$\EN_{\min}/N$. For $P\ge\overline{\EN}/N$, we can achieve the maximum
possible capacity $C_{\delta}(P)=1$. For powers lower than $\EN_{\min}/N$, capacity goes to
zero. So, the interesting range of calculation if for $NP\in(\EN_{\min},\overline{\EN})$.
\subsection{Diagonally-dominant channels}
A matrix $A=(a_{ij})$ is said to be row-diagonally-dominant or simply
diagonally-dominant if $|a_{ii}|\ge\sum_{j,j\ne i}|a_{ij}|$. Let us call channels $h$ for which the matrix $(M_hM_h^T)^{-1}$ exists
and is diagonally-dominant as diagonally-dominant channels. For such
channels, the minimum energy values $\EN(s)$ can be characterized as follows.
\begin{lemma}
\label{lem:one}
When the matrix $ (M_h M_h^T)^{-1}$ is row-diagonally-dominant,
$\EN(s)$ for $s\in\{-1,1\}^N$ is achieved at $x^*$ that satisfies the equality constraints
\begin{align}
\diag(s)M_hx^* =\delta\i.
\label{eq:one}
\end{align}
\end{lemma}
\begin{proof}
Since the constraint set is linear, the optimization problem
\eqref{eq:conv_ineq} is strongly dual. 
We now solve the optimization problem by forming its dual. The Lagrangian is given by
\begin{align*}
L(\lambda, x) = \|x\|^2+\lambda^T(\delta \mathbf{1} - \diag(x)M_h x).
\end{align*}
The gradient of the Lagrangian with respect to $x$ is given by $2x- (\diag(x) M_h)^T \lambda$, which gives
\[x^* = \frac{(\diag(x) M_h)^T \lambda}{2}.\]
Substituting $x^*$ in the Lagrangian, the dual problem is given by
\begin{align}
&\max_ {\lambda \succeq 0} \underbrace{\frac{-\lambda^T\diag(x)M_hM_h^T\diag(x)\lambda}{4} +\delta\lambda^T\mathbf{1}}_{g(\lambda)}.
\end{align}
We have
\begin{align}
\nabla g(\lambda) = \frac{-\diag(s)M_hM_h^T\diag(s)\lambda}{2} +\delta \mathbf{1}.
\label{eq:lag}
\end{align}
Setting the gradient to zero in \eqref{eq:lag}, we obtain the optimal $x^*$ as the solution to 
\[ {-\diag(s)M_h x^*} +\delta \mathbf{1}=0,\]
which is exactly  equivalent to \eqref{eq:one}. The only caveat is that the $\lambda$ obtained  from $g(\lambda)=0$  should be in the positive orthant.  Solving $g(\lambda)=0$,  we obtain 
\[\lambda^* =  2\delta   \diag(s) (M_h
M_h^T)^{-1}\diag(s)\mathbf{1}.\] 
We want $\lambda^*\succeq 0$ for every $s$. We can easily observe that this is indeed true if $ (M_h M_h^T)^{-1}$ is diagonally-dominant. 
\end{proof}
% {\color{red} Show that $(1,\epsilon)$ channel satisfies the diagonal dominance criterion.}
%{\color{red} Give examples of other channels that satisfy this constraint.}
Hence, from Lemma \ref{lem:one}, when $ (M_h M_h^T)^{-1}$ is diagonally-dominant, $x^*=\delta M_h^{-1}s$. Hence, 
\begin{equation}
  \label{eq:1}
  \EN(s) =\|x^*\|^2= \delta^2s^T G s,
\end{equation}
where $G = (M_hM_h^T)^{-1}$. 

In the next lemma, the mean energy $\overline{\EN}$ is
characterized in terms of the channel matrix $M_h$ for
diagonally-dominant channels.
 \begin{lemma}
 \label{lem:mean}
The mean energy for diagonally-dominant channels is given by 
\begin{align*}
\overline{\EN}&=  \delta^2 \trace( M_h^{-T} M_h^{-1}).
\end{align*}
\end{lemma}
\begin{proof}
The energy $\EN(s)$ in \eqref{eq:1} can be expanded as
\begin{align}
\EN(S) =\delta^2\sum_{i=1}^N G_{ii} +\sum_{i,j,i\ne j} G_{ij}s_is_j.
\end{align}
Hence,
\begin{align*}
 &\frac{1}{2^N}\sum_{s\in\{-1,1\}^N}\EN(s)=\\
 &\frac{\delta^2}{2^N}\sum_{s\in\{-1,1\}^N} \sum_{i=1}^N G_{ii}
 +\underbrace{ \sum_{i\ne j} G_{ij} \frac{1}{2^N}\sum_{s\in\{-1,1\}^N}s_i s_j}_{=0}
\end{align*}
The second term in the above sum is zero since the summation spans over all the sequences on $\{1, -1\}^N$. 
Hence
\begin{align*}
 \overline{\EN}=  \delta^2\sum_{i=1}^N G_{ii}\stackrel{(a)}{=}  \delta^2\trace( M_h^{-T} M_h^{-1}),
\end{align*}
where $(a)$ follows from the definition of the matrix $G$.
\end{proof}
We now characterize $\overline{P}_h\triangleq\lim_{N\to\infty}\overline{\EN}/N$,
which is the minimum average power needed for capacity of 1 bit, in terms of the
Fourier transform of the channel $h$. Let the discrete-Fourier
transform of the channel be
\begin{equation}
  \label{eq:2}
  f(\lambda) = \sum_{k=0}^{L-1}h_k e^{jk\lambda}.
\end{equation}
Since $M_h$ is a circulant matrix,  it is easy to see that
 \begin{align}
\overline{\EN}=\delta^2\trace( M_h^{-T} M_h^{-1})=\delta^2\sum_{k=1}^N \frac{1}{|f\left(\frac{2\pi k}{N}\right)|^2}.
\label{eq:pz}
\end{align}
Using standard arguments,  it can be shown that \cite{gray}
\begin{align}
\overline{P}_h \to  \frac{\delta^2}{2\pi}\int_0^{2\pi}\frac{1}{|f(\lambda)|^2}\d \lambda.
\label{eq:pz_inf}
\end{align}
Observe that $\overline{P}_h$ is the energy of the inverse of the channel
scaled by $\delta^2$, and is related to the power needed for zero-forcing.  

In summary, for diagonally-dominant channels, the approximate ISI
capacity is given by
\begin{align*}
C_{\delta}(P) = \begin{cases}
1&\text{if } P\ge \overline{P}_h,\\
 \beta P+\ln(Z)&  \text{if }  \underline{P}_h\le P\le \overline{P}_h,\\
0&\text{if } P< \underline{P}_h,
\end{cases}
 \end{align*}
 where $\underline{P}_h=\lim_{N\to\infty}\EN_{\min}/N$ (assuming limit
 exists) and $Z$ is the normalizing constant for the Gibbs distribution. 

\section{Achievable Schemes}
\label{sec:achievable-schemes}
We now consider achievable schemes for the approximate ISI channel
under that assumption that the channel matrix $M_h$ is invertible. In
achievable schemes, an information sequence $B\in\{-1,1\}^N$ with a
well-chosen distribution is
encoded into a channel input $x$ that satisfies $|y_n|\ge\delta$. The
rate of transmission over the approximate ISI channel is $H(B)/N$.
\subsection{Zero-forcing with Gibbs distribution}
 Let $b=\{b_1, b_2,\hdots, b_N\} \in \{-1,1\}^N$ be an instance of the
 information sequence $B$. Choose the input to the channel as 
\begin{align}
x = \delta M_h^{-1} b,
\label{eq:opt}
\end{align}
which implies that the output of the ISI channel is $y=M_h x = \delta
b$. This is referred to as zero-forcing because it
involves channel inversion. Hence, the output of the quantizer $s$ equals the information sequence $b$. 
When the channel is  diagonally dominant, Lemma \ref{lem:one} and
Section \ref{sec:gibbs} imply that a Gibbs distribution on $B$ and the
choice $x =\delta M_h^{-1} s$  as the input to the channel results in
a capacity-achieving scheme. Hence, for diagonally-dominant $h$, the scheme in \eqref{eq:opt} is optimal when $b$ is sampled from the  Gibbs distribution given in \eqref{eq:gibbs}.  

 However, $b$ is a sequence of length $N$, and it is well known that sampling from a Gibbs distribution has exponential complexity in $N$.  For achieving capacity,  $N$ should be very large which makes this scheme impractical. In the next subsection, we compute the entropy rate when $b$ is sampled from a Markov chain instead of a Gibbs distribution. 
\subsection{Zero-forcing with Markov input}
\label{sec:MS}
 As before we choose $x = \delta M_h^{-1} b$, where $b$ is the
 information sequence. The sequence  $b$ is sampled from a two state
 Markov chain shown in Fig. \ref{fig:MC} with transition matrix 
\begin{eqnarray}
	 P = \left( \begin{array}{ccc}
	 \alpha  & 1-\alpha   \\
     1-\alpha & \alpha  \end{array} \right),
     \label{eq:mark}
\end{eqnarray}  
where $0\leq\alpha\leq 1$.
Observe that  $s_n = b_n$ and the achievable rate of this scheme
equals the entropy  $H(B)=H_2(\alpha)$, where $H_2(x)$  is the binary
entropy function. The average transmit power, denoted
$P_{zm}(\alpha)$, is given by 
\begin{align*}
P_{zm}(\alpha) &= \frac{1}{N}\E[\|X\|^2] = \frac{\delta^2}{N}\E[B^TM_h^{-T} M_h^{-1}B]\\
&= \frac{\delta^2}{N} \trace( M_h^{-T} \E[BB^T]M_h^{-1}) .
\end{align*}
\begin{figure}
\centering 
\begin{tikzpicture}[->,>=stealth',shorten >=1pt,auto,node distance=4cm,
 			     	thick,main node/.style={circle,draw,
			     	font=\sffamily,minimum size=10mm}]
  					\node[main node] (M) {1};
  					\node[main node] (F) [right of = M] {-1};
  					\path[every node/.style={font=\sffamily\large,
  					fill=white,inner sep=2pt}]
  					(M) edge [loop left] node[left=1mm]{$\alpha$} (M) 
      				edge [bend left=60] node[yshift=1mm] {$1-\alpha$} (F)     
  					(F) edge [loop right] node[right=1mm]{$\alpha$} (F) 
      				edge [bend right=-60] node[yshift=-1mm] {$1-\alpha$} (M);
			\end{tikzpicture} 
\caption{Markov chain with transition matrix given in \eqref{eq:mark} }
\label{fig:MC}
\end{figure}
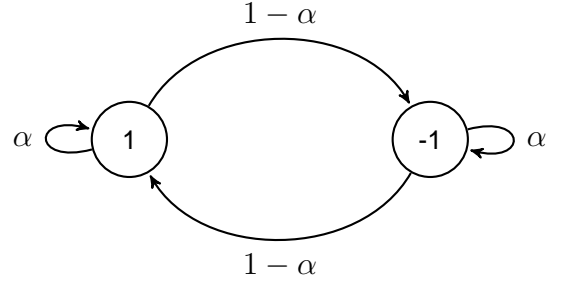
Using the eigenvalue decomposition of $P$, it can be shown that 
 \[
	 P^{d} = \left( \begin{array}{ccc}
	 \frac{1 + (2\alpha -1)^{d}}{2}  & \frac{1 - (2\alpha -1)^{d}}{2} \\
     \frac{1 - (2\alpha -1)^{d}}{2}  & \frac{1 + (2\alpha -1)^{d}}{2} \end{array}\right).
\]
 Using $(0.5, 0.5)$  as stationary distribution and $P^d$,  we get
\begin{equation}
\mathbb{E}[b_n b_{n+d}] = \beta^d,
\end{equation}
where $\beta =2\alpha -1$. 
Hence the correlation matrix $R=\E[BB^T]$ is given by $R_{ij} = \beta^{|j-i|}$. So for finite $N$
\begin{align*}
P_{zm}(\alpha) &= \frac{\delta^2}{N} \trace(  R M_h^{-1}M_h^{-T}).
\end{align*}
When $N$ is large, the  Topelitz matrix $R$ can be approximated by a circulant matrix \cite{gray}  and 
\begin{align}
P_{zm}(\alpha) &\to  \frac{\delta^2}{2\pi}\int_0^{2\pi}\frac{1}{|f(\lambda)|^2}\left(\frac{2(1-\beta\cos(\lambda))}{1+\beta^2-2\beta\cos(\lambda)} -1\right)\d \lambda.
\label{eq:pzm}
\end{align}
The value of $\alpha$ is chosen so as to maximize the entropy
rate and the rate achieved at power $P$, denoted $\mathcal{R}_m(P)$,
is obtained as  
\begin{align}
\mathcal{R}_m(P) = \max_{\alpha: P_{zm}(\alpha)\leq P} H_2(\alpha).
 \end{align}
\section{Numerical examples}
\label{sec:numerical-examples}
In this section, we evaluate the approximate ISI capacity and the rate
achieved by the Markov scheme in Section \ref{sec:MS} for some sample
channels. For numerical evaluation, we assume $\delta =0.3$. 
 
The channel $(1,\epsilon), |\epsilon|<1$ is a diagonally-dominant
channel and zero forcing with Gibbs distribution is an optimal strategy. For this channel, 
 \[f(\lambda) = 1+\epsilon e^{j\lambda}.\] 
Using \eqref{eq:pz_inf}, the minimum power required for zero-forcing is given by
\[\overline{P}_h= \frac{\delta^2}{2\pi}\int_0^{2\pi} \frac{1}{ 1+\epsilon^2+2\epsilon \cos(\lambda)}\d \lambda = \frac{\delta^2}{1-\epsilon^2}.\]
While difficult to prove theoretically, by careful simulations, it can
be observed that the minimum energy $\EN_{\min}$  is obtained for the
sequences $\pm(1,1,\hdots,1)$. Using this observation,
$\underline{P}_h=\lim_{N\to\infty}\EN_{\min}/N$ is obtained as
\[\underline{P}_h =\frac{\delta^2}{(1+\epsilon)^2}.\]
The transmit power required for the two state Markov scheme \eqref{eq:pzm} is 
 \[P_{zm}(\alpha) =  \frac{\delta ^2}{1-\epsilon ^2} \left[\frac{1+ \epsilon (1-2\alpha)}{1- \epsilon (1-2\alpha)}\right].\]
 Hence the  maximum entropy problem for the two state Markov chain translates to 
 \begin{align*}
&\mathcal{R}_m(P)=\max H_2(\alpha),\\
\text{such that}\quad  &\alpha > \frac{1}{2} + \frac{1}{2 \epsilon} \frac{1 - P\delta^{-2}(1-\epsilon ^2)}{1 + P\delta^{-2}(1-\epsilon ^2)}.
 \end{align*}
The solution of the above problem is given by 
\begin{align*}
\mathcal{R}_m(P) = 
\begin{cases}
1&\text{if }   P\geq \overline{P}_h\\
 H_2\left( \frac{1}{2} + \frac{1}{2 \epsilon} \frac{1 - P\delta^{-2}(1-\epsilon ^2)}{1 + P\delta^{-2}(1-\epsilon ^2)}\right)  &\text{if }\frac{\delta^2}{(1+\epsilon)^2}\leq P<\overline{P}_h\\
0  &\text{if } P<\frac{\delta^2}{(1+\epsilon)^2}.
\end{cases}
 \end{align*}
 We first observe that  the  capacity $C_{\delta}(P)$ and the achievable rate $\mathcal{R}_m(P) $ match at $\underline{P}_h$ and $\overline{P}_h$ being equal to $0$ and $1$, respectively. 
\begin{figure}
\centering
\begin{tikzpicture}
\begin{axis}[
	xlabel={Normalised energy: P/$\delta ^2$},
	ylabel={bits per channel use},
	line width = 1pt, 
	width=0.4\paperwidth, height=0.26\paperheight,xmax=1.05,xmin=0.65,ymin=0,ymax=1,xmajorgrids,ymajorgrids,legend pos= south east]
	
	\addplot[smooth,color=blue,mark=square] plot coordinates
	{
   ( 0.6944,    0.0834)
   ( 0.7178,    0.2686)
   ( 0.7411,    0.4073)
   ( 0.7644,    0.5231)
   ( 0.7867,    0.6205)
   ( 0.8100,    0.7019)
   ( 0.8333,    0.7695)
   ( 0.8567,    0.8255)
   ( 0.8800,    0.8715)
   ( 0.9033,    0.9090)
   ( 0.9256,    0.9389)
   ( 0.9489,    0.9621)
   ( 0.9722,    0.9793)
   ( 0.9956,    0.9910)
   ( 1.0189,    0.9978)
   ( 1.0422,    0.9999)

	};
	\addplot[smooth,color=blue,mark=star] plot coordinates
	{
    (0.6944,         0)
    (0.7178,    0.2399)
    (0.7411,    0.3949)
    (0.7644,    0.5159)
    (0.7867,    0.6146)
    (0.8100,    0.6962)
    (0.8333,    0.7642)
    (0.8567,    0.8207)
    (0.8800,    0.8676)
    (0.9033,    0.9059)
    (0.9256,    0.9366)
    (0.9489,    0.9607)
    (0.9722,    0.9784)
    (0.9956,    0.9907)
    (1.0189,    0.9978)
    (1.0422,    0.9999)
};	
    \addplot[dashed,color=red] plot coordinates
    { 
    (0.6944,    0.0834)
    (0.7178,    0.5627)
    (0.7411,    0.7413)
    (0.7644,    0.8311)
    (0.7867,    0.8870)
    (0.8100,    0.9230)
    (0.8333,    0.9471)
    (0.8567,    0.9636)
    (0.8800,    0.9753)
    (0.9033,    0.9836)
    (0.9256,    0.9895)
    (0.9489,    0.9939)
    (0.9722,    0.9968)
    (0.9956,    0.9986)
    (1.0189,    0.9996)
    (1.0422,    0.9999)
};
    \addplot[dashed,color=red,mark=star] plot coordinates
    {   
    (0.6944,         0)
    (0.7178,    0.3503)
    (0.7411,    0.5420)
    (0.7644,    0.6732)
    (0.7867,    0.7675)
    (0.8100,    0.8369)
    (0.8333,    0.8881)
    (0.8567,    0.9259)
    (0.8800,    0.9535)
    (0.9033,    0.9731)
    (0.9256,    0.9864)
    (0.9489,    0.9946)
    (0.9722,    0.9989)
    (0.9956,    0.9999)
    (1.0189,    0.9983)
    (1.0422,    0.9947)
};

\legend{$C(P)$$\epsilon = 0.2$, $ \mathcal{R}_m(P)$ $\epsilon = 0.2$, $C(P)$ $\epsilon = 0.8$,$ \mathcal{R}_m(P)$ $\epsilon = 0.8$ }
\title{$-\epsilon$,$1$}
\end{axis}
\end{tikzpicture}
\caption{$\mathcal{R}_m(P)$ and $C(P)$ versus normalised energy $P/\delta ^2$ for $ \{1, \epsilon \}$  channel}
\label{fig:cap1}
\end{figure}
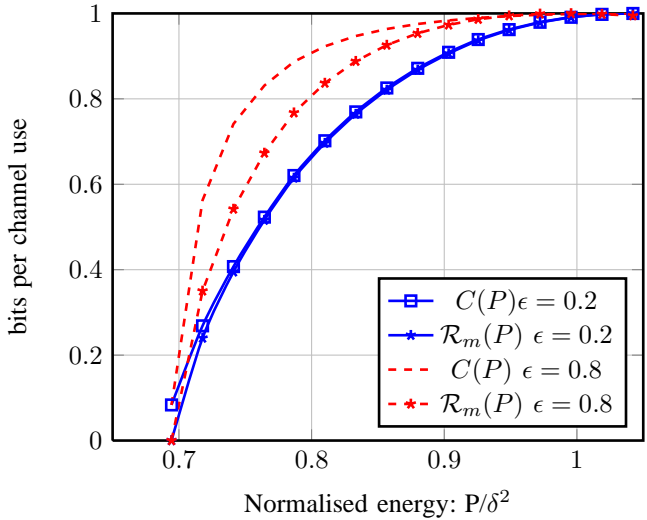
In Fig. \ref{fig:cap1}, the approximate capacity and the achievable rate of the Markov scheme are plotted as a function of normalized power $P/\delta^2$ for $\epsilon=0.2$ and $0.8$.  We observe that $\mathcal{R}_m(P) $ is very close to capacity for $\epsilon=0.2$ and the gap increases with $\epsilon$.  
 
In Fig. \ref{fig:cap2}, $\mathcal{R}_m(P) $ and $C_{\delta}(P)$ are plotted for a non-diagonally dominant channel $h=(-0.3, 1, 0.6)$. The energies $\EN(s)$  are obtained  by numerically solving \eqref{eq:conv_ineq}. For this channel $\overline{P}_h/\delta^2 \approx 0.838$ and $\underline{P}_h/\delta^2 \approx 0.56$ and $\mathcal{R}_m(P) =0$ for $P/\delta^2 < 0.59$. 
\begin{figure}
\centering
\begin{tikzpicture}
\begin{axis}[
	xlabel={Normalised energy: P/$\delta ^2$},
	ylabel={bits per channel use},
	line width = 1pt, 
	width=0.4\paperwidth, height=0.26\paperheight,xmax=0.85,xmin=0.55,ymin=0,ymax=1,xmajorgrids,ymajorgrids,legend pos= south east]
	
	\addplot[smooth,color=red,mark=square] plot coordinates
	{
    (0.5644,         0)
    (0.5744,    0.5386)
    (0.5856,    0.6456)
    (0.5967,    0.7177)
    (0.6078,    0.7700)
    (0.6189,    0.8098)
    (0.6300,    0.8414)
    (0.6400,    0.8672)
    (0.6511,    0.8887)
    (0.6622,    0.9069)
    (0.6733,    0.9223)
    (0.6844,    0.9354)
    (0.6956,    0.9468)
    (0.7056,    0.9565)
    (0.7167,    0.9649)
    (0.7278,    0.9721)
    (0.7389,    0.9781)
    (0.7500,    0.9833)
    (0.7611,    0.9877)
    (0.7711,    0.9911)
    (0.7822,    0.9940)
    (0.7933,    0.9963)
    (0.8044,    0.9979)
    (0.8156,    0.9991)
    (0.8267,    0.9998)
    (0.8367,    0.9999)

	};
	\addplot[smooth,color=blue,mark=star] plot coordinates
	{
    (0.5922,    0.0003)
    (0.6022,    0.2735)
    (0.6133,    0.4510)
    (0.6244,    0.5686)
    (0.6356,    0.6531)
    (0.6467,    0.7177)
    (0.6578,    0.7690)
    (0.6678,    0.8105)
    (0.6789,    0.8448)
    (0.6900,    0.8734)
    (0.7011,    0.8974)
    (0.7122,    0.9176)
    (0.7233,    0.9346)
    (0.7333,    0.9489)
    (0.7444,    0.9608)
    (0.7556,    0.9706)
    (0.7667,    0.9789)
    (0.7778,    0.9854)
    (0.7889,    0.9906)
    (0.7989,    0.9944)
    (0.8100,    0.9972)
    (0.8211,    0.9991)
    (0.8322,    0.9999)
};	
\legend{$C(P)$, $ \mathcal{R}_m(P)$}
\title{$-\epsilon$,$1$}
\end{axis}
\end{tikzpicture}
\caption{$\mathcal{R}_m(P)$ and $C(P)$ versus normalised energy $P/\delta ^2$ for $\{-\epsilon, 1, 2\epsilon \}$ channel with $\epsilon$ = 0.3}
\label{fig:cap2}
\end{figure}
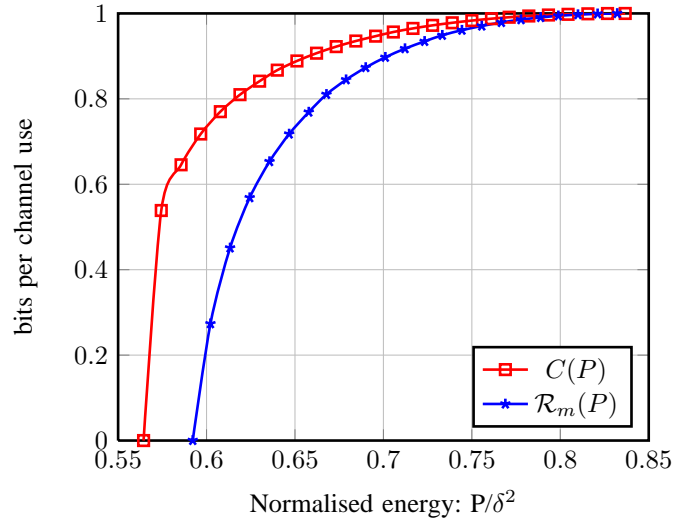
We observe that at higher powers the Markov scheme is close to
capacity even for this non-diagonally dominant example. 

\section{Concluding Remarks}
\label{sec:concluding-remarks}
The capacity of the approximate output-quanitzed ISI channel is characterized using
Gibbs distribution, and Markov schemes are shown to approach
capacity. The characterization is complete for the case of
diagonally-dominant channels. Extensions to more general channels and
better achievable schemes are interesting problems for future
study. Another important problem is bounding the error in the
information rate because of the approximation, which is complicated by
the dependencies introduced by the Gibbs distribution. 

\bibliographystyle{IEEEtran} %{plain} {usrt}
\bibliography{IEEEabrv,isirefs} %{ITgeneral,mu}

\end{document}